\documentclass[nolimit]{article}
\usepackage{amsmath,amssymb,amsthm,mathrsfs,amssymb,mathrsfs,txfonts}
\newtheorem{thm}{Theorem}[section]

\newtheorem{lem}[thm]{Lemma}

\theoremstyle{definition}

\begin{document}
\date{}
\title{\bf  Amending coherence-breaking channels via unitary operations}
\author{ \\Long-Mei Yang$^1$, Bin Chen$^2$, Tao Li$^3$,  Shao-Ming Fei$^1$, Zhi-Xi Wang$^1$
%\thanks{Corresponding author: wangzhx@cnu.edu.cn}
\\
{\footnotesize $^1$School of Mathematical Sciences, Capital Normal University, Beijing 100048, China}\\
{\footnotesize $^2$School of Mathematical Sciences, Tianjin Normal University, Tianjin 300387, China}\\
{\footnotesize $^3$School of Science, Beijing Technology and Business University, Beijing 100048, China}}
\maketitle
\begin{abstract}
The coherence-breaking channels play a significant role in quantum information theory.
We study the coherence-breaking channels and give a method to amend the
coherence-breaking channels by applying unitary operations.
For given incoherent channel $\Phi$, we give necessary and sufficient conditions for the channel to be a coherence-breaking channel and amend it via unitary operations.
For qubit incoherent channels  $\Phi$ that are not coherence-breaking ones, we consider the mapping $\Phi\circ\Phi$ and present the conditions for coherence-breaking and channel amendment as well.
\end{abstract}


\section{Introduction}
Originating from quantum superposition, quantum coherence has been a cornerstone of quantum information theory.
It is of fundamental importance in quantum information processing such as quantum reference frames \cite{SDB,IM,RW},
transport in biological systems \cite{SL,NL} and quantum thermodynamics \cite{JA,PC}.
In recent years, the resource theories of quantum coherence have been rapidly developed \cite{TB,BY,EC}.
The free operations, the free states  and the resource states are three basic ingredients in a quantum resource theory.
In the resource theory of quantum coherence, the free states are the incoherent states whose
density matrices are diagonal under the reference basis.
The free operations are the incoherent operations $\Phi$.

As a hot topic, channel properties related to theoretical and experimental work has attracted many attention \cite{coherent,QIP,Long1,Long2,evolution,Long3}.
An important aspect in the study of channel properties is related to the evolution under the action of a channel.
For example, in the entanglement resource theory, entanglement-breaking channels (EBCs) have been characterized completely \cite{EBCs,QEB,EB}.
In Ref. \cite{A.C}, Cuevas {\it et al.} amended EBCs by using unitary operations.
Of special interest to us is the coherence-breaking channel and its amendment.
Recently, in Ref. \cite{CBCs}, Bu {\it et al.} introduced two kinds of coherence-breaking channels (CBCs).
In addition, they devoted to the coherence-breaking indices of incoherent quantum channels and presented various examples to elucidate this concept.
For a given quantum channel $\Phi$, if $\Phi^k=\underbrace{\Phi\circ\cdots\circ\Phi}\limits_{k \  {\rm times}}$ is a CBC while $\Phi^{k-1}$ is not,
then the coherence-breaking index $n(\Phi)$ of $\Phi$ is $k$, i.e., $n(\Phi)=k$.

In this paper, we investigate the structure of quantum channels and give necessary and sufficient conditions for a channel to be CBC,
in particular, for qubit quantum channels $\Phi$ with $n(\Phi)=2$.
Furthermore, we consider the amendment of the corresponding CBCs via unitary operations. Detailed examples related to qubit CBCs's amendment are presented.

\section{Preliminaries}
We first introduce some relevant basic concepts.
For a $d$-dimensional quantum system and a fixed reference basis $\{|i\rangle\}$,
the $l_1$ norm of coherence $\mathcal{C}_{l_1}$ of a state $\rho$ is given by
$\mathcal{C}_{l_1}(\rho)=\sum\limits_{i\neq j}|\langle i|\rho|j\rangle|$.
Any qubit state $\rho$ can be written as $\rho=\frac{1}{2}(I+\overrightarrow{r}\cdot\overrightarrow{\mathbf{\sigma}})=
\frac{1}{2}(I+r_x\sigma_x+r_y\sigma_y+r_z\sigma_z)$,
where $\sigma_x,\sigma_y$ and $\sigma_z$ are Pauli matrices,
$\vec{r}$ is a 3-dimensional Bloch vector with $|\overrightarrow{r}|\leq 1$.

A quantum channel $\Phi$ is a linear completely positive and trace preserving (CPTP) map \cite{Nielson}.
The action of a qubit quantum channel $\Phi$ on a state $\rho$ can be expressed by
a real $4\times 4$ matrix $\left(
                                                                \begin{array}{cc}
                                                                  1 & \mathbf{0}_{1\times 3} \\
                                                                  \overrightarrow{n} & M \\
                                                                \end{array}
                                                              \right)$
which transforms the column vector $(1,r_x,r_y,r_z)^T$ to the corresponding one of $\Phi(\rho)$, where $\mathbf{0}_{1\times 3}=(0,0,0)$,
$\overrightarrow{n}$ is a real 3-dimensional column vector, $M$ is a $3\times 3$ real matrix, $T$ stands for transpose.
Obviously, the action of a qubit quantum channel $\Phi$ is completely characterized by $(M,\overrightarrow{n})$
and the $k$-th iterated channel $\Phi^k$ is characterized by $(M^k,(\sum\limits_{i=0}^{k-1}M^i)\overrightarrow{n})$.

A non-coherence-generating channel (NC) $\tilde{\Phi}$  is a CPTP map from an incoherent state to an incoherent state:
$\tilde{\Phi}(\mathcal{I})\subset \mathcal{I}$, where $\mathcal{I}$ denotes the set of incoherent states \cite{NC}.
Any quantum channel $\Phi$ is called an incoherent channel if there exists a Kraus decomposition
$\Phi(\cdot)=\Sigma_iK_i(\cdot)K_i^{\dagger}$ such that $\rho_i=\frac{K_i(\rho)K_i^{\dagger}}{{\rm Tr}(K_i(\rho)K_i^{\dagger})}$
is incoherent for any incoherent state $\rho$.
We call incoherent channel $\Phi$ a coherence-breaking channel (CBC) if $\Phi(\rho)$  is an incoherent state for any state $\rho$ \cite{CBCs}.

Let $\Phi$ be an incoherent channel, the coherence-breaking index $n(\Phi)$ of $\Phi$ is defined as \cite{CBCs}
\begin{equation}
n(\Phi)=\min\{k\geq1: \text{$\Phi^{k}$ is a coherence-breaking channel}\}.
\end{equation}

A rank-$2$ qubit channel is an NC if and only if it has the Kraus decomposition either as \cite{NC}
\begin{equation}\label{eqnc1}
\Phi^{(1)}(\cdot)=E_1^{(1)}(\cdot)E_1^{(1)\dagger}+E_2^{(1)}(\cdot)E_2^{(1)\dagger}
\end{equation}
with
\begin{equation}\label{eqNC1}
E_1^{(1)}=\left(
        \begin{array}{cc}
          e^{i\eta}\cos\theta\cos\phi & 0 \\
          -\sin\theta\sin\phi & e^{i\xi}\cos\phi \\
        \end{array}
      \right), \ \ \
     E_2^{(1)}=\left(
\begin{array}{cc}
                  \sin\theta\cos\phi & e^{i\xi}\sin\phi \\
                  e^{-i\eta}\cos\theta\sin\phi & 0 \\
                \end{array}
              \right)
\end{equation}
or as
\begin{equation}\label{eqnc2}
\Phi^{(2)}(\cdot)=E_1^{(2)}(\cdot)E_1^{(2)\dagger}+E_2^{(2)}(\cdot)E_2^{(2)\dagger}
\end{equation}
 with
\begin{equation}\label{eqNC2}
E_1^{(2)}=\left(
        \begin{array}{cc}
         \cos\theta & 0 \\
         0 & e^{i\xi}\cos\phi \\
        \end{array}
      \right),  \ \ \
      E_2^{(2)}=\left(
                \begin{array}{cc}
                  0 & \sin\phi \\
                  e^{i\xi}\sin\theta & 0 \\
                \end{array}
              \right),
\end{equation}
where $\theta,\phi,\xi$ and $\eta$ are all real numbers.
Here we note that $\Phi^{(1)}$ is not an incoherent channel unless
$\sin\theta\cos\theta\sin\phi\cos\phi=0$ and $\Phi^{(2)}$ is an incoherent channel.

\section{Coherence-breaking channels}
We first give a way to characterize the structure of $d$-dimensional (qudit) quantum channels and present the necessary and sufficient condition for the channels to be CBCs.
For qubit quantum channels, Bu {\it et al.} have already provided a necessary and sufficient condition to determine
whether a given channel is a CBC or not.
Below we give new and more specific necessary  and sufficient conditions to justify a given channel based on the results derived by
Hu {\it et al.} \cite{NC}.

Any $d$-dimensional state $\rho$ can be expressed as \cite{RAB}:
\begin{equation*}
\rho=\frac{I}{d}+\frac{1}{2}\sum\limits_{j=0}^{d-2}\sum\limits_{k=j+1}^{d-1}(b_s^{jk}\sigma_s^{jk}+b_a^{jk}\sigma_a^{jk})+
\frac{1}{2}\sum\limits_{l=1}^{d-1}b^{l}\sigma^{l}
\end{equation*}
with $b_{s(a)}^{jk}={\rm tr} (\rho\sigma_{s(a)}^{jk})$, $b^{l}={\rm tr} (\rho\sigma^{l})$, $\sigma_s^{jk}=|j\rangle\langle k|+|k\rangle\langle j|$,
$\sigma_a^{jk}=-i|j\rangle\langle k|+i|k\rangle\langle j|$ and
$\sigma^{l}=\sqrt{\frac{2}{l(l+1)}}(\sum\limits_{j=0}^{l-1}|j\rangle\langle j|-l|l\rangle\langle l|)$.
Then the state $\rho$ can be rewritten as $\rho=\frac{I}{d}+\vec{x}\cdot\vec{X}$ with
$\vec{x}=(b_s^{01},\ b_a^{01},\ldots,\\ b_s^{d-2,d-1},\ b_a^{d-2,d-1},\ b^1,\ldots, b^{d-1})^{T}$
and $\vec{X}=(\sigma_s^{01},\ \sigma_a^{01},\ldots,\sigma_s^{d-2,d-1}, \sigma_a^{d-2,d-1},\ \sigma^1,\ldots,\sigma^{d-1})^{T}$.
Thus, the general action of a $d$-dimensional quantum channel $\Phi$ on a qudit state $\rho$ can be expressed by a real $d^2\times d^2$ matrix
$\left(
        \begin{array}{cc}
          p & \vec{m}^{T} \\
          \vec{n} & M \\
        \end{array}
      \right)$, where $\vec{m}, \ \vec{n}$ are $(d^2-1)\times 1$ column vectors and $M$ is a $(d^2-1)\times(d^2-1)$ matrix.
Then $\Phi=\left(
        \begin{array}{cc}
          1 & \vec{0}^{T} \\
          \vec{n} & M \\
        \end{array}
      \right)$
since $\Phi$ maps $\rho=\frac{I}{d}+\vec{x}\cdot\vec{X}$ to $\Phi(\rho)=\frac{I}{d}+\vec{x^{\prime}}\cdot\vec{X^{\prime}}$,
where $\vec{x^{\prime}}=(b_s^{\prime01},\ b_a^{\prime01},\ldots,b_s^{\prime d-2,d-1},\ b_a^{\prime d-2,d-1},\ b^{\prime1},
\ldots,b^{\prime d-1})^{T}$
and $\vec{X^{\prime}}=(\sigma_s^{\prime01},\ \sigma_a^{\prime01},\ldots,\sigma_s^{\prime d-2,d-1}, \sigma_a^{\prime d-2,d-1},\ \sigma^{\prime1},\\
\ldots,\sigma^{\prime d-1})^{T}$.

\begin{lem}\label{qudits}
A $d$-dimensional quantum channel $\Phi$ represented by $(M,\vec{n})$ is a CBC if and only if $M$ and $\vec{n}$
have the following form:
\begin{equation}\label{qudit}
M=\left(
        \begin{array}{cccc}
          0 & 0 & \cdots & 0\\
          \vdots & \vdots & \vdots & \vdots\\
          0 & 0 & \cdots & 0\\
          M_{d^2-d+1,1} & M_{d^2-d+1,2} & \cdots & M_{d^2-d+1,d^2-1}\\
          \vdots & \vdots & \vdots & \vdots\\
          M_{d^2-1,1} & M_{d^2-1,2} & \cdots & M_{d^2-1,d^2-1}\\
        \end{array}
      \right), \ \ \ \
\vec{n}=\left(
        \begin{array}{c}
          0 \\
          \vdots\\
          0 \\
          n_{d^2-d+1}\\
          \vdots\\
          n_{d^2-1}\\
        \end{array}
      \right).
\end{equation}
\end{lem}
\begin{proof}
Obviously we can see that $b_{s,a}^{jk}=0$ for all $0\leq j,k \leq d-1$ if and only if $\rho\in\mathcal{I}$.
Thus, $M$ and $\vec{n}$ are of the above forms.
\end{proof}

For the case $d=2$, Lemma \ref{qudits} reduces to the Proposition 1 given in \cite{CBCs}.
For qubit quantum channels, we give below some more specific conditions to justify whether the channel is a CBC or not.

\begin{lem}\label{Amend1}
Let $\Phi$ be an incoherent channel defined by \eqref{eqnc1}, then $\Phi$ is a CBC if and only if $\cos\theta=0$.
\end{lem}
\begin{proof}
 Any density operator acting on a two-dimensional quantum system can be generally written as
\begin{equation}\label{eqState}
\rho=\left(
        \begin{array}{cc}
          a & b \\
          b^* & 1-a \\
        \end{array}
      \right),
\end{equation}
where $|a|^2+|b|^2\leq1$.
Substituting \eqref{eqState} into \eqref{eqnc1}, we have
\begin{equation}
\Phi(\rho)=\left(
                       \begin{array}{cc}
                         A & B\\
                         B^* & 1-A\\
                       \end{array}
                     \right),
\end{equation}
where $A=a\cos^2\phi+(b^*e^{i\xi}+be^{-i\xi})\sin\theta\sin\phi\cos\phi+(1-a)\sin^2\phi$ and
$B=be^{i\eta-i\xi}\cos\theta\cos^2\phi+b^*e^{i\xi+i\eta}\cos\theta\sin^2\phi$.
Then we find that $\Phi(\rho)\in\mathcal{I}$ if and only if $B=be^{i\eta-i\xi}\cos\theta\cos^2\phi+b^*e^{i\xi+i\eta}\cos\theta\\
\sin^2\phi=0$ for arbitrary $b$.
Let $b=|b|e^{\beta}$, where $\beta\in [0,2\pi)$.
Thus, $B=0$ if and only if $e^{-i\eta}|b|\cos\theta \\ (e^{i\xi-i\beta}\cos^2\phi+e^{i\beta-i\xi}\sin^2\phi)=0$ for arbitrary $\beta$ and $|b|$.
Particularly, for $\beta=\xi$, we have that $\cos\theta=0$.
It is easy to see that $\Phi$ is an incoherent channel when $\cos\theta=0$.
\end{proof}

\begin{lem}
Let $\Phi$ be an incoherent channel defined by \eqref{eqnc2}, then $\Phi$ is a CBC if and only if $\sin\theta=0, \ \cos\phi=0$ or $\cos\theta=0, \ \sin\phi=0$.
\end{lem}
\begin{proof}
Substituting \eqref{eqState} into \eqref{eqnc2}, we have
\begin{equation}
\Phi(\rho)=\left(
                       \begin{array}{cc}
                         C & D\\
                         D^* & 1-C\\
                       \end{array}
                     \right),
\end{equation}
where $C=a\cos^2\theta+(1-a)\sin^2\phi$ and $D=e^{i\xi}(b\cos\theta\cos\phi+b^*\sin\theta\sin\phi)$.
Then we find that $\Phi$ is a CBC if and only if $D=0$ for arbitrary $b$.
Let $b=|b|e^{\beta}$.
Then $\Phi$ is a CBC if and only if $|b|\sqrt{\cos^2\beta\cos^2(\theta-\phi)+\sin^2\beta\cos^2(\theta+\phi)}=0$
for any $|b|$ and $\beta$, which gives rise to that either $\sin\theta=0, \ \cos\phi=0$ or $\cos\theta=0, \ \sin\phi=0$.
\end{proof}

%A general discussion for qudit quantum channels being CBCs with $n(\Phi)=2$ can be very difficult.
%Here, we only discuss the qubit quantum channels.

For qubit channels $\Phi$ with $n(\Phi)=2$, we have the following necessary and sufficient conditions for $\Phi^2$ to be CBC.

\begin{lem}\label{2(Phi)}
Let $\Phi$ be an incoherent channel defined by \eqref{eqnc1}. Then $n(\Phi)=2$ if and only if $\cos2\phi=0, \ \sin\theta=0$ and
$\sin\xi\sin\eta+\cos\xi\cos\eta=0$.
\end{lem}
\begin{proof}
It is easy to see that $\sin\theta\cos\theta\sin\phi\cos\phi=0$.
Comparing \eqref{eqnc1},\ \eqref{eqNC1} with \eqref{qudit} for $d=2$, we find that
$n_x=n_y=n_z=M_{13}=M_{23}=0$, $M_{11}=\cos\theta(\cos\eta\cos\xi+\sin\eta\sin\xi\cos2\phi)$, $M_{12}=\cos\theta(\sin\eta\cos\xi\cos2\phi-\cos\eta\sin\xi)$,
$M_{21}=\cos\theta(\cos\eta\sin\xi\cos2\phi-\sin\eta\cos\xi)$, $M_{22}=\cos\theta(\cos\eta\\
\cos\xi\cos2\phi+\sin\eta\sin\xi)$,
$M_{31}=2\sin\theta\sin\phi\cos\phi\cos\xi$, $M_{32}=-2\sin\theta\sin\phi\cos\phi\sin\xi$ and $M_{33}=\cos2\phi$.
Thus, $\Phi^2$ is a CBC if and only if $M_{11}^2+M_{12}M_{21}=0$, $M_{12}(M_{11}+M_{22})=0$, $M_{21}(M_{11}+M_{22})=0$,
$M_{12}M_{21}+M_{22}^2=0$ and $\cos\theta\neq0$ if and only if $\cos2\phi=0, \ \cos\theta\neq0$ and $\sin\xi\sin\eta+\cos\xi\cos\eta=0$.
Then we have $n(\Phi)=2$ if and only if $\cos2\phi=0, \ \sin\theta=0$ and $\sin\xi\sin\eta+\cos\xi\cos\eta=0$.
\end{proof}

\begin{lem}\label{CBC2}
Let $\Phi$ be an incoherent channel defined by \eqref{eqnc2}. Then $n(\Phi)=2$ if and only if one of the following three conditions holds:

{\rm (i)} $\cos2\xi=0$ and $\cos(\theta+\phi)=\cos(\theta-\phi)\neq0$;

{\rm (ii)} $\cos(\theta-\phi)=0, \ \cos\xi=0$ and $\cos(\theta+\phi)\neq0$;

{\rm (iii)} $\cos(\theta+\phi)=0, \ \cos\xi=0$ and $\cos(\theta-\phi)\neq0$.
\end{lem}
\begin{proof}
Comparing \eqref{eqnc2}, and \eqref{eqNC2} with \eqref{qudit}, we have that $n_x=n_y=M_{13}=M_{23}=M_{31}=M_{32}=0$, $n_z=\sin^2\phi-\sin^2\theta$,
 $M_{11}=\cos\xi\cos(\theta-\phi)$, $M_{12}=\sin\xi\cos(\theta+\phi)$,
 $M_{21}=-\sin\xi\cos(\theta-\phi)$, $M_{22}=-\cos\xi\cos(\theta+\phi)$ and $M_{33}=\cos^2\theta-\sin^2\phi$.
Thus, $n(\Phi)=2$ if and only if $M_{11}^2+M_{12}M_{21}=0$, $M_{12}(M_{11}+M_{22})=0$, $M_{21}(M_{11}+M_{22})=0$,
$M_{12}M_{21}+M_{22}^2=0$, $\sin\theta\cos\phi\neq0$ and $\cos\theta\sin\phi\neq0$,
which implies that one of the above three conditions holds.
\end{proof}

\section{Amending coherence-breaking channels}
In this section, we discuss the amendment of CBCs.
All the channels we discussed in this section is represented by $(M,\vec{n})$ with at least one of
the conditions $M\neq0$ and $\vec{n}\neq\vec{0}$ holds.
We show that a CBC $\Phi$ can be amended via unitary operations $\Lambda_\alpha$ through $\Lambda_\alpha\circ\Phi$.
For the case $n(\Phi)=2$, we show the channels can always be amended by unitary operations through $\Phi\circ\Lambda_\alpha\circ\Phi$
while $\Lambda_\alpha\circ\Phi^2$ may not always amend $\Phi^2$.

\begin{lem}
For a $d$-dimensional CBC $\Phi$, there always exists a unitary operation $\Lambda_\alpha$ that amends $\Phi$ by $\Lambda_\alpha\circ\Phi$.
\end{lem}
\begin{proof}
If $\Phi$ is a CBC, we have $\Phi(\rho)=\rho^{\prime}=(\rho_{00}^{\prime},\ \rho_{11}^{\prime},\ldots,\ \rho_{d-1,d-1}^{\prime})^{T}\in\mathcal{I}$.
If $d$ is even, let
\begin{equation}
\Lambda_\alpha(\cdot)=\left(
                                     \begin{array}{cccc}
                                          \tilde{\Lambda}_\alpha &  0 & \cdots & 0 \\
                                          0 & \tilde{\Lambda}_\alpha & \cdots & 0\\
                                          \vdots & \vdots & \ddots & \vdots\\
                                          0 & 0 & 0 & \tilde{\Lambda}_\alpha \\
                                        \end{array}
                                      \right)(\cdot)\left(
                                     \begin{array}{cccc}
                                          \tilde{\Lambda}_\alpha^{T} &  0 & \cdots & 0 \\
                                          0 & \tilde{\Lambda}_\alpha^{T} & \cdots & 0\\
                                          \vdots & \vdots & \ddots & \vdots\\
                                          0 & 0 & 0 & \tilde{\Lambda}_\alpha^{T} \\
                                        \end{array}
                                      \right)
\end{equation}
be a unitary operation with $\tilde{\Lambda}_\alpha=\left(
                                     \begin{array}{cc}
                                          \cos\alpha & -\sin\alpha\\
                                          \sin\alpha & \cos\alpha\\
                                        \end{array}
                                      \right)$.
Then
$\Lambda_\alpha(\rho^{\prime})=diag(\Lambda_{\alpha_0}, \Lambda_{\alpha_2}, ..., \Lambda_{\alpha_{d-2}})$
with \\
$\Lambda_{\alpha_i}=\left(
                                     \begin{array}{cc}
                                          \rho_{ii}^{\prime}\cos^\alpha+\rho_{i+1,i+1}^{\prime}\sin^2\alpha & (\rho_{ii}^{\prime}-\rho_{i+1,i+1}^{\prime})\sin\alpha\cos\alpha\\
                                          (\rho_{ii}^{\prime}-\rho_{i+1,i+1}^{\prime})\sin\alpha\cos\alpha & \rho_{ii}^{\prime}\sin^\alpha+\rho_{i+1,i+1}^{\prime}\cos^2\alpha\\
                                        \end{array}
                                      \right)$, $i=0,1,...,d-2$.
Therefore, the channel $\Phi$ can be amended if $\sin2\alpha\neq0$.

If $d$ is odd, let
\begin{equation}
\Lambda_\alpha(\cdot)=\left(
                                     \begin{array}{ccccc}
                                          1 & 0 & 0 & \cdots & 0 \\
                                          0 & \tilde{\Lambda}_\alpha & 0 & \cdots &0 \\
                                          0 & 0 & \tilde{\Lambda}_\alpha & \cdots &0 \\
                                          \vdots & \vdots & \vdots & \ddots &\vdots \\
                                          0 & 0 & 0 & \cdots &\tilde{\Lambda}_\alpha \\
                                        \end{array}
                                      \right)(\cdot)\left(
                                     \begin{array}{ccccc}
                                          1 & 0 & 0 & \cdots & 0 \\
                                          0 & \tilde{\Lambda}_\alpha^T & 0 & \cdots &0 \\
                                          0 & 0 & \tilde{\Lambda}_\alpha^T & \cdots &0 \\
                                          \vdots & \vdots & \vdots & \ddots &\vdots \\
                                          0 & 0 & 0 & \cdots &\tilde{\Lambda}_\alpha^T \\
                                        \end{array}
                                      \right)
\end{equation}
 be a unitary operation. Then $\Lambda_\alpha(\rho^{\prime})=diag(\rho_{00}^{\prime},\Lambda_{\alpha_0},\Lambda_{\alpha_2},...,\Lambda_{\alpha_{d-2}})$.
Thus, the channel can be amended if $\sin2\alpha\neq0$.
\end{proof}

\begin{lem}\label{Amend}
Let $\Phi$ be a CBC defined by \eqref{eqnc1}.
There always exists a unitary operation $\Lambda_\alpha$ that amends $\Phi$ by $\Lambda_\alpha\circ\Phi$.
\end{lem}
\begin{proof}
Any general unitary operation can be written as
\begin{equation*}
\Lambda_\alpha(\cdot)=\left(
                        \begin{array}{cc}
                          \cos\alpha & -e^{i\alpha_1}\sin\alpha \\
                          e^{i\alpha_2}\sin\alpha & e^{i\alpha_3}\cos\alpha \\
                        \end{array}
                      \right)(\cdot)\\
                      \left(
                        \begin{array}{cc}
                          \cos\alpha & e^{-i\alpha_2}\sin\alpha \\
                          -e^{-i\alpha_1}\sin\alpha & e^{-i\alpha_3}\cos\alpha \\
                        \end{array}
                      \right).
\end{equation*}
Comparing it with \eqref{qudit} for $d=2$,
we get
\begin{equation}
\Lambda_\alpha=\left(
\begin{array}{cccc}
1 & 0 & 0 & 0 \\
0 & N_{11} & N_{12} & N_{13} \\
 0 & N_{21} & N_{22} & N_{23} \\
 0 & N_{31} & N_{32} & N_{33} \\
 \end{array}
 \right),
\end{equation}
where $N_{11}=\cos^2\alpha\cos\alpha_3-\sin^2\alpha\cos(\alpha_1-\alpha_2)$, $N_{12}=\sin^2\alpha\sin(\alpha_1-\alpha_2)-\cos^2\alpha\sin\alpha_3$,
$N_{13}=\sin2\alpha\cos\alpha_2$, $N_{21}=\cos^2\alpha\sin\alpha_3+\sin^2\alpha\sin(\alpha_1-\alpha_2)$,
$N_{22}=\sin^2\alpha\cos(\alpha_1-\alpha_2)+\cos^2\alpha\cos\alpha_3$, $N_{23}=\sin2\alpha\sin\alpha_2$,
$N_{31}=-\sin2\alpha\cos\alpha_1$, $N_{32}=\sin2\alpha\sin\alpha_1$
and $N_{33}=\cos2\alpha$.
Then
\begin{equation}
\Lambda_\alpha\circ\Phi=\left(
\begin{array}{cccc}
1 & 0 & 0 & 0 \\
0 & N_{13}M_{31} & N_{13}M_{32} & N_{13}M_{33} \\
 0 & N_{23}M_{31} & N_{23}M_{32} & N_{23}M_{33} \\
 0 & N_{33}M_{31} & N_{33}M_{32} & N_{33}M_{33} \\
 \end{array}
 \right),
\end{equation}
where $M_{ij}$ are defined in Lemma \ref{Amend1} with $i,j=1,2,3$.
Thus, $\Lambda\circ\Phi$ is not a CBC if and only if $N_{13}\neq0$ or $N_{23}\neq0$, namely, $\sin2\alpha\neq0$.
In other words, $\Lambda$ can amend $\Phi$ if and only if $\sin2\alpha\neq0$.
\end{proof}

Similarly, we can prove the following results:

\begin{lem}
Let $\Phi$ be a CBC defined by \eqref{eqnc2}.
There always exists a unitary operation $\Lambda_\alpha$ that amends $\Phi^2$ by $\Lambda_\alpha\circ\Phi$.
\end{lem}

%For the case $n(\Phi)=2$, we only discuss qubit quantum channels, for high dimension channels,
%it can be very difficult to give a general discussion since the expression of $\Phi$ is quite complex.

\begin{lem}\label{amend1}
Let $\Phi$ be a CBC defined by \eqref{eqnc1} with $n(\Phi)=2$.
There always exists a unitary operation $\Lambda_\alpha$ that amends $\Phi^2$ by $\Phi\circ\Lambda_\alpha\circ\Phi$.
\end{lem}
\begin{proof}
It is direct to see that
\begin{equation}
\Phi\circ\Lambda_\alpha\circ\Phi=\left(
\begin{array}{cccc}
1& 0 & 0 & 0 \\
0 & \tilde{M}_{11} & \tilde{M}_{12} &  0\\
0 & \tilde{M}_{21} & \tilde{M}_{22} & 0 \\
0 & \tilde{M}_{31} & \tilde{M}_{32} & 0 \\
  \end{array}
                          \right),
\end{equation}
where
$\tilde{M}_{11}=M_{11}(M_{11}N_{11}+M_{12}N_{21})+M_{21}(M_{11}N_{12}+M_{12}N_{22})+M_{31}(M_{11}N_{13}+M_{12}N_{23})$,
$\tilde{M}_{12}=M_{12}(M_{11}N_{11}+M_{12}N_{21})+M_{22}(M_{11}N_{12}+M_{12}N_{22})+M_{32}(M_{11}N_{13}+M_{12}N_{23})$,
$\tilde{M}_{21}=M_{11}(M_{21}N_{11}+M_{22}N_{21})+M_{21}(M_{21}N_{12}+M_{22}N_{22})+M_{31}(M_{21}N_{13}+M_{22}N_{23})$,
$\tilde{M}_{22}=M_{12}(M_{21}N_{11}+M_{22}N_{21})+M_{22}(M_{21}N_{12}+M_{22}N_{22})+M_{32}(M_{21}N_{13}+M_{22}N_{23})$,
$\tilde{M}_{31}=M_{11}(M_{31}N_{11}+M_{32}N_{21})+M_{21}(M_{31}N_{12}+M_{32}N_{22})+M_{31}(M_{31}N_{13}+M_{32}N_{23})$,
$\tilde{M}_{32}=M_{12}(M_{31}N_{11}+M_{32}N_{21})+M_{22}(M_{31}N_{12}+M_{32}N_{22})+M_{32}(M_{31}N_{13}+M_{32}N_{23})$,
and $M_{ij}$ and $N_{ij}$ are defined in Lemmas \ref{2(Phi)} and \ref{Amend}, respectively.

Assume that $\sin\alpha=0$. We have $\tilde{M}_{ij}=0$ if and only if $\sin\alpha_3=0$, where $i=1,2$ and $j=1,2,3$.
Thus, $\Phi^2$ is amended if $\sin\alpha=0$ and $\sin\alpha_3\neq0$.
\end{proof}

\begin{lem}
Let $\Phi$ be a CBC defined by \eqref{eqnc2} and $n(\Phi)=2$.
There always exists a unitary operation $\Lambda_\alpha$ that amends $\Phi^2$ by $\Phi\circ\Lambda_\alpha\circ\Phi$.
\end{lem}
\begin{proof}
Similar to Lemma \ref{amend1}, we obtain
\begin{equation}
\Phi\circ\Lambda_\alpha\circ\Phi=\left(
\begin{array}{cccc}
1& 0 & 0 & 0 \\
\tilde{n}_x & \tilde{M}_{11} & \tilde{M}_{12} &  0\\
\tilde{n}_y & \tilde{M}_{21} & \tilde{M}_{22} & 0 \\
\tilde{n}_z & \tilde{M}_{31} & \tilde{M}_{32} & 0 \\
  \end{array}
                          \right),
\end{equation}
where $\tilde{n}_x=n_z(M_{11}N_{13}+M_{12}N_{23})$, $\tilde{n}_y=n_z(M_{21}N_{13}+M_{22}N_{23})$, $\tilde{n}_z=n_zM_{33}N_{33}$,
$\tilde{M}_{11}=M_{11}(M_{11}N_{11}+M_{12}N_{21})+M_{21}(M_{11}N_{12}+M_{12}N_{22})$,
$\tilde{M}_{12}=M_{12}(M_{11}N_{11}+M_{12}N_{21})+M_{22}(M_{11}N_{12}+M_{12}N_{22})$,
$\tilde{M}_{13}=M_{33}(M_{11}N_{13}+M_{12}N_{23})$,
$\tilde{M}_{21}=M_{11}(M_{21}N_{11}+M_{22}N_{21})+M_{21}(M_{21}N_{12}+M_{22}N_{22})$,
$\tilde{M}_{22}=M_{12}(M_{21}N_{11}+M_{22}N_{21})+M_{22}(M_{21}N_{12}+M_{22}N_{22})$,
$\tilde{M}_{23}=M_{33}(M_{21}N_{13}+M_{22}N_{23})$,
$\tilde{M}_{31}=M_{11}M_{33}N_{31}+M_{21}M_{33}N_{32}$,
$\tilde{M}_{32}=M_{12}M_{33}N_{31}+M_{22}M_{33}N_{32}$, $\tilde{M}_{33}=M_{33}^2N_{33}$,
and $M_{ij},\ N_{ij}$ are given in Lemmas \ref{CBC2} and \ref{Amend}, respectively.

{\rm (i)} $\cos2\xi=0$ and $\cos(\theta+\phi)=\cos(\theta-\phi)\neq0$.

Assume  $\sin\alpha=0$. Then $\tilde{M}_{ij}=0$ if and only if $\sin\alpha_3=0$ for $i=1,2$ and $j=1,2,3$.
Thus, $\Phi^2$ is amended if $\sin\alpha=0$ and $\sin\alpha_3\neq0$.

{\rm (ii)} $\cos(\theta-\phi)=0, \ \cos\xi=0$ and $\cos(\theta+\phi)\neq0$.

In this case,
\begin{equation}
\Phi\circ\Lambda_\alpha\circ\Phi=\left(
\begin{array}{cccc}
1& 0 & 0 & 0 \\
\tilde{n}_x & 0 &M^2_{12}N_{21} &  M_{12}M_{33}N_{23}\\
\tilde{n}_y & 0 & 0 & 0 \\
\tilde{n}_z & 0 & 0 & 0 \\
  \end{array}
                          \right).
\end{equation}
Hence, $\Phi\circ\Lambda\circ\Phi$ is not a CBC if and only if $N_{21}\neq0$ or $N_{23}\neq0$.

{\rm (iii)} $\cos(\theta+\phi)=0, \ \cos\xi=0$ and $\cos(\theta-\phi)\neq0$.

Similar to {\rm (ii)}, for $N_{12}\neq0$ or $N_{13}\neq0$, the channel $\Phi$ can be amended.
\end{proof}

Although a CBC $\Phi$ can always be amended via unitary operations by $\Lambda_\alpha\circ\Phi$ for the case $n(\Phi)=1$,
there may exist $\Phi$ that cannot be amended by $\Lambda_\alpha\circ\Phi^2$ for the case $n(\Phi)=2$.
Authors in \cite{error} pointed out that in some case, channel recovery, which is a duality quantum operation described in
\cite{Long4,Long5}, can be done via a linear combination of unitary operations.
However, in the case $n(\Phi)=2$, $\Phi$ also can not be amended by the linear combination of unitary operations through $L_c\circ\Phi^2$
with $L_c=\sum\limits_{\alpha}c_\alpha\Lambda_\alpha$,
if $\Phi$ cannot be amended by $\Lambda_\alpha\circ\Phi^2$,
 where $c_i$ is a complex number and $U_i$ are all unitary,
and $c_\alpha^{ \ ,}s$ satisfy $\sum\limits_{\alpha}|c_\alpha|\leq1$.

Now we give examples to illustrate our results about the coherence-breaking channel's amendment.

$\mathbf{Example\  4.1}$
Consider an incoherent qubit quantum channel $\Phi$ characterized by $(M,\vec{n})$ with
$M=\left(
   \begin{array}{ccc}
     0 & 0 & 0 \\
     0 & 0 & 0 \\
     0 & 0 & \mu \\
   \end{array}
 \right)$, where $\mu$ is a  real number and $\vec{n}=(0,0,0)^T$.
 For $|\mu|\leq1$, $\Phi$ is an incoherent channel \cite{MBR,CK}.
 It is easy to see that $\Phi$ is a CBC.
 Consider unitary operation $\Lambda_\alpha(\cdot)=\left(
                        \begin{array}{cc}
                          \cos\alpha & -\sin\alpha \\
                          \sin\alpha & \cos\alpha \\
                        \end{array}
                      \right)(\cdot)
                      \left(
                        \begin{array}{cc}
                         \cos\alpha & \sin\alpha \\
                          -\sin\alpha & \cos\alpha \\
                        \end{array}
                      \right)$
 with $\sin2\alpha\neq0$.
We obtain that $\Lambda_\alpha\circ\Phi$ is not a CBC, i.e., the channel $\Phi$ is amended.

$\mathbf{Example \ 4.2}$
Consider an incoherent qubit quantum channel $\Phi$ characterized by $(M,\vec{n})$ with
$M=\left(
   \begin{array}{ccc}
     0 & \gamma & 0 \\
     0 & 0 & 0 \\
     0 & 0 & 0 \\
   \end{array}
 \right)$, where $\gamma$ is a  real number and $\vec{n}=(0,0,0)^T$.
 For $|\gamma|\leq1$, $\Phi$ is an incoherent channel \cite{MBR,CK}.
 It is easy to see that $n(\Phi)=2$.
Consider unitary operation $\Lambda_\alpha(\cdot)=\left(
                        \begin{array}{cc}
                          1 & 0 \\
                          0 & e^{i\alpha_3} \\
                        \end{array}
                      \right)(\cdot)\
                      \left(
                        \begin{array}{cc}
                          1 & 0 \\
                          0 & e^{-i\alpha_3}\\
                        \end{array}
                      \right)$
 with $\sin\alpha_3\neq0$.
One finds that $\Phi\circ\Lambda_\alpha\circ\Phi$ is not a CBC, i.e., the channel is amended.
But the channel cannot be amended by $\Lambda_\alpha\circ\Phi^2$ via unitary operations since $\Lambda_\alpha\circ\Phi^2(\rho)=\frac{I}{2}$.

\section{Conclusions}
We have discussed the structure of CBCs and their amendments.
Particularly, for qubit CBCs, we have given the specific expressions for the case of $n(\Phi)=1,2$ and shown that the channels can be amended via
unitary operations $\Lambda_\alpha$ by $\Lambda_\alpha\circ\Phi$ and $\Phi\circ\Lambda_\alpha\circ\Phi$, respectively.
For $n(\Phi)\geq3$, following similar discussions, there may also exist unitary operations to amend the channel $\Phi$ by $\underbrace{\Phi\circ\Lambda_\alpha\circ\Phi\circ\cdots\circ\Phi\circ\Lambda_\alpha\circ\Phi}\limits_{n \  \Phi}$.
In addition, for a $d$-dimensional quantum channel $\Phi$ with $n(\Phi)\geq 2$, one may conjecture that the channel $\Phi$ can also be amended by
$\underbrace{\Phi\circ\Lambda_\alpha\circ\Phi\circ\cdots\circ\Phi\circ\Lambda_\alpha\circ\Phi}\limits_{n \  \Phi}$
with proper unitary operations $\Lambda_\alpha$.

\section{Acknowledgements}
This work is supported by the NSFC  (11675113) and the Research Foundation for Youth Scholars of Beijing Technology and Business University
 (QNJJ2017-03) and the Scientific Research General Program of Beijing Municipal Commission of Education  (Grant No. KM201810011009).

\end{document}